\documentclass[
final
]{dmtcs-episciences}

\usepackage[utf8]{inputenc}
\usepackage{subfigure}

\usepackage[round]{natbib}

\usepackage{arcs}

\usepackage[usenames,dvipsnames]{xcolor}
\usepackage{makeidx} 
\usepackage{algorithm}
\usepackage{algorithmic}
\usepackage{graphicx,tipa}
\usepackage{arcs,lmodern,fix-cm}
\usepackage{amsmath,amstext,amsthm,url,setspace, enumerate}
\usepackage{rotating}

\usepackage{fourier}

\newcommand{\barr}[1]{\overline{#1}}
\newcommand{\vect}[1]{\overrightarrow{#1}}
\newcommand{\sinn}[1]{\sin \left({#1}\right)}

\newcommand{\coss}[1]{\cos \left({#1}\right)}
\newcommand{\RA}{\ensuremath{R_1}}
\newcommand{\RB}{\ensuremath{R_2 }}
\newcommand{\postB}{\ensuremath{t_d}} 
\newcommand{\postC}{\ensuremath{t_d}} 
\def\qed{\hfill\rule{2mm}{2mm}}
\newcommand{\ff}{face-to-face }
\newcommand{\arccc}[1]{
\widearc{#1}
}
\newcommand{\ignore}[1]{}

\newtheorem{theorem}{Theorem}[section]
\newtheorem{claim}[theorem]{Claim}

\newtheorem{lemma}[theorem]{Lemma}

\theoremstyle{definition}

\author{
J. Czyzowicz\affiliationmark{1}~\thanks{Research supported in part by NSERC of Canada.} 
\and
K. Georgiou\affiliationmark{2}~\thanks{Research supported in part by NSERC of Canada.}
\and 
E. Kranakis\affiliationmark{3}~\thanks{Research supported in part by NSERC of Canada.}
\and
L. Narayanan\affiliationmark{4}~\thanks{Research supported in part by NSERC of Canada.} \\
\and 
J. Opatrny\affiliationmark{4}~\thanks{Research supported in part by NSERC of Canada.}
\and
B. Vogtenhuber\affiliationmark{5}
}

\title[Evacuating Robots from a Disk Using Face-to-Face Communication]{Evacuating Robots from a Disk \\
Using Face-to-Face Communication\footnote{An extended abstract of this paper appeared in the proceedings of the 8th International Conference on Algorithms and Complexity (CIAC'15)~\cite{CzyzowiczGKNOV15}.
}
}
\affiliation{
  D\'{e}pt. d'Informatique, Universit\'{e} du Qu\'{e}bec en Outaouais,
Gatineau, QC, Canada\\
  Dept. of Mathematics, 
Ryerson University, 
Toronto, ON, Canada\\
School of Computer Science, Carleton University, Ottawa ON, Canada\\
Dept. of Comp. Science and Soft. Engineering, 
Concordia University, Montreal, QC,  Canada \\
Institute of Software Technology,
Graz University of Technology, Graz, Austria}
\received{2020-03-13}
\revised{2020-06-12, 2020-07-22}
\accepted{2020-07-31}
\begin{document}
\publicationdetails{22}{2020}{4}{4}{6198}
\maketitle
\begin{abstract}
Assume that two robots are located at the centre of a unit disk.
Their goal is to {\em evacuate} from the disk 
through an {\em exit} at an unknown location on the boundary of the disk.
At any time the robots can move anywhere they choose on the disk, independently of each other, with maximum speed $1$. 
The robots can cooperate by exchanging information 
whenever they meet.
We study algorithms for the two robots to minimize the {\em evacuation time}: the time when {\em both} robots reach the exit.

Czyzowicz et al.~(2014)
gave an algorithm defining trajectories for the two 
robots yielding evacuation time at most $5.740$ and also proved that any algorithm has evacuation time at least $3+ \frac{\pi}{4} + \sqrt{2} \approx 5.199$.
We improve both the upper and lower bound on the evacuation time of a unit disk. 
Namely, we present a new non-trivial algorithm whose evacuation time is at most $5.628$
and show that any algorithm has evacuation time at least $3+ \frac{\pi}{6} + \sqrt{3} \approx 5.255$.
To achieve the upper bound, we designed an algorithm which proposes a forced meeting between the two robots, even if the exit has not been found by either of them. We also show that such a strategy is provably optimal for a related problem of searching for an exit placed at the vertices of a regular hexagon. 
\end{abstract}

\section{Introduction}

The goal of traditional search problems is to find an object that is
located in a specific domain. This subject of research has a long history
and there is a plethora of models 
investigated in the mathematical and theoretical
computer science literature with emphasis
on probabilistic search in~\cite{stone1975theory},
game theoretic applications in~\cite{alpern2003theory},
cops and robbers in~\cite{anthony2011game},
classical pursuit and evasion in~\cite{nahin2012chases},
search problems and group testing in ~\cite{ahlswede1987search},
and many more. 

In this paper, we investigate the problem of searching for a stationary point target called an {\em exit} at an unknown location using two robots. 
This type of collaborative search is advantageous in that it reduces the required search time by distributing the search effort between the two robots. 
In previous work on collaborative search, the goal has generally been to minimize the time taken by the {\em first} robot to find the object of the search. 
In contrast, in the work at hand, we are interested in minimizing the time when the {\em last robot} finds the exit. 
In particular, suppose two robots are in the interior of a region with a single exit. 
The robots need to evacuate the region but the location of the exit is unknown to them. 
The robots can cooperate to search for the exit, but it is not enough for one robot to find the exit, we require {\em both} robots to reach  the exit as soon as possible. 

We study the problem of two robots that  start at the same time at the centre of a unit disk and attempt to reach an exit placed  at an unknown location on the boundary of the disk. 
At any time the robots can move anywhere they choose within the disk. 
Indeed, they can take short-cuts by moving in the interior of the disk if desired. 
We assume that their maximum speed is 1. 
The robots can communicate with each other only if they are at the same point at the same time: we call this communication model {\em \ff  communication}.  
Our goal is to schedule the trajectories of the robots so as to minimize the {\em evacuation time}, which is the time it takes both robots to reach the exit (for the worst case location of the exit). 

Our main contributions pertain to improved upper and lower bounds for the evacuation problem on the disc. The main ideas are derived by the study of a related evacuation problem in which the exit lies in one of the vertices of a regular hexagon. For the latter problem, we introduce a novel and non-intuitive search strategy in which searchers are forced to detour and meet if the exit is not found early enough. Surprisingly, we show this trajectory to induce optimal evacuation cost for that problem. The implications are two-fold. First, the lower bound we obtain applies directly to the main evacuation problem on the disc. Second, we use the same ideas of a detour and forced meeting to improve the best evacuation cost known (at the time when the result was first announced; see Section~\ref{sec: related work} for newer developments).

\subsection{Related work}
\label{sec: related work}


The most related work to ours is ~\cite{CGGKMP}, where the evacuation problem for a set of robots all
starting from the centre of a unit disk was introduced and studied.  Two communication models are introduced in \cite{CGGKMP}. In the {\em wireless} model,
the two robots can communicate at any time regardless of their locations. In particular, a robot that finds the exit can immediately communicate its location to the other
robot. The other model is called the  {\em non-wireless or local} model in \cite{CGGKMP}, and is the same as our \ff model:  two robots can only communicate when they are face to face,  that is, they are at the same point location at the same time. In \cite{CGGKMP}, for the case of 2 robots, an algorithm with  evacuation time $1 +\frac{2\pi}{3} + \sqrt{3} \approx 4.826$ is  given for the wireless model; this is shown to be optimal. For the \ff  model,  they prove an upper bound of 
$5.740$ and a lower bound of  $5.199$ on the evacuation time. 
Since the first announcement of our results in~\cite{CzyzowiczGKNOV15}, 
Brandt at al.~\cite{Watten2017} improved our upper bound from 5.628 to 5.625. Their contributions also pertain to a significant simplification of the search trajectories (that use detours but without forced meetings), along with a clever performance analysis. 
Very recently, \cite{disser2019evacuating} reported a further improvement of 5.6234 by employing and generalizing techniques introduced in~\cite{Watten2017}.
However, to the best of our knowledge, the lower bound we provide in this manuscript is the best currently known. 

Since the introduction of the problem in~\cite{CGGKMP}, a number of interesting variations emerged some of which are summarized in recent survey~\cite{CGK19search} (see also~\cite{flocchini2019distributed} for a collection of broadly related problems).
Some representative examples of related problems include 
an evacuation problem with speed bounds in the wireless model~\cite{lamprou2016fast},
evacuation from triangles~\cite{CzyzowiczKKNOS15,ChuangpishitMNO17},
evacuation from multiple rays~\cite{BrandtFRW20},
evacuation from graphs~\cite{Borowiecki0DK16},
search with terrain dependent speeds~\cite{CzyzowiczKKNOS17},
wireless evacuation from multiple exits~\cite{czyzowicz2018evacuating,PattanayakR0S18},
priority evacuation~\cite{CGKKKNOS18b,CzyzowiczGKKKNO18},
evacuation with faulty robots~\cite{czyzowicz2017evacuation,pattanayak2019chauffeuring,GKLPP19Algosensors} or with probabilistically faulty robots~\cite{bgmp2020probabilistically},	
evacuation with immobile agents~\cite{GKKa16,dmtcs:5528},
time energy tradeoffs for evacuation on the line \cite{czyzowiczICALP2019,kranakis2009time},
worst-case average-case tradeoffs for evacuation on the disc~\cite{chuangpishit2018average}, 
and searching graphs~\cite{angelopoulos2019expanding}.

In a different direction, Baeza-Yates {\em et al} posed the question of minimizing the worst-case trajectory of a single robot searching for a target point at an unknown location in the plane ~\cite{baezayates1993searching}. This was generalized to multiple robots in \cite{LS01}, and more recently has been studied  in \cite{Emekicalp2014,Lenzen2014}. However, in these papers, the robots cannot communicate, and moreover, the objective is for the first robot to find the target. 
Two seminal and influential papers (that appeared almost at the
same time) on probabilistic search that concern minimizing the {\em expected time}  for the robot to find the target are
\cite{beck1964linear}~and~\cite{bellman1963optimal}.
As for two surveys on search theory the interested reader is referred to~\cite{benkoski1991survey}~and~\cite{dobbie1968survey}.
The latter survey also contains 
an interesting classification of
search problems by search objectives, distribution of effort, point target
(stationary, large, moving), two-sided search, and other criteria.
The evacuation problem considered in our paper
is related to searching on a line, in that we are searching
on the boundary of a disk but with the
additional ability to make short-cuts in order
to enable the robots to meet sooner
and thus evacuate faster. 
Our problem is also related to the  {\em rendezvous problem}  and the problem of  {\em gathering} \cite{alpern1999asymmetric,PGNP2005}. Indeed our problem can be seen as a version of a rendezvous problem for three robots, where one of them remains stationary.

\subsection{Preliminaries and notation}
\label{subsec:prelim}

We assume that two robots \RA\ and \RB\  are initially at the centre of a disk with radius 1, and that there is an exit at some location $X$ on the boundary of the disk. 
The robots do not know $X$, but do know each other's algorithms. 
The robots move at a speed subject to a maximum speed, say 1. 
They cannot communicate except if they are at the same location at the same time. 
The {\em evacuation problem} is to define trajectories for the two robots that minimize the {\em evacuation time}.
All our results (upper and lower bounds) pertain to deterministic algorithms that work against a (worst case) deterministic adversary.

For two points $A$ and $B$ on the unit circle, the length of an arc $AB$ is denoted by $\arccc{AB}$, 
while the length of the corresponding chord (line segment) will be denoted by $\overline{AB}$ 
(arcs on the circle are always read clockwise, i.e., arc $AB$ together with arc $BA$ cover the whole circle). 
By $\angle{ABC}$ we denote the angle at $B$ in the triangle $ABC$. Finally by $\vect{AB}$ we denote the vector with initial point $A$ and terminal point $B$, .i.e. vector $B-A$ when $B,A$ are understood as points in the Cartesian plane.

\subsection{Outline, results of the paper and new improvements}

In \cite{CGGKMP} an algorithm is given defining a
trajectory for two robots in the \ff communication
model with evacuation time
$5.740$ and it is also proved that any such algorithm has evacuation 
time at least $3+ \frac{\pi}{4} + \sqrt{2} >  5.199$.

Our main contribution in this paper is to improve both the upper and lower bounds
on the evacuation time. Namely,
we give a new algorithm whose evacuation time is at most $5.628$
(see Section~\ref{sec:Evacuation Algorithms})
and also prove that any algorithm has evacuation 
time at least $3+ \frac{\pi}{6} + \sqrt{3} > 5.255$
(see Section~\ref{sec:Lower Bound}).
To prove our lower bound on the disk, we first give tight bounds for the problem of evacuating a regular hexagon where the exit is placed at an unknown vertex. 
We observe that, surprisingly, in our optimal evacuation algorithm for the hexagon, the two robots are forced to meet after visiting a subset of vertices, even if an exit has not been found at that time. 
We use the idea of such a forced meeting in the design of our disk evacuation algorithm, inducing this way an improvement of the upper bound from 5.740 to 5.628. 
It is still unknown  whether such a forced meeting is necessary. However, since the first announcement of the upper bound, Brandt at al.~\cite{Watten2017} proposed an elegant analysis of a simplified trajectory that avoids forced meetings, still improving the upper bound from 5.628 to 5.625. The authors in the latter paper even proposed a refinement of their technique involving multiple detours as the robots are searching for the exit. Their ideas were later materialized in~\cite{disser2019evacuating} that reported a further improvement of 5.6234.


\section{Evacuation Algorithms}
\label{sec:Evacuation Algorithms}

In this section we give two new evacuation algorithms for two robots in the  \ff model that take evacuation time approximately
$5.644$ and $5.628$ respectively. 
We begin by presenting  Algorithm $\mathcal A$ proposed by \cite{CGGKMP}
which has been shown to have evacuation time 
$5.740$.
Our goal is to understand the worst possible configuration for this algorithm, and to subsequently modify it accordingly so as to improve its performance. 

 All the algorithms we present follow the same general structure: 
 The two robots \RA\ and \RB\ start by moving together to an arbitrary point $A$ on the boundary of the disk. 
 Then \RA\ explores the arc $A'A$, where $A'$ is the antipodal point of $A$, by moving along some trajectory defined by the algorithm.
 At the same time, \RB\ explores the arc $AA'$, following a trajectory that is the reflection of \RA's trajectory.
 If either of the robots finds the exit, it immediately uses the {\em Meeting Protocol} defined below to meet the other robot (note that the other robot has not yet found the exit and hence keeps exploring). 
 After meeting, the two robots travel together on the shortest path to the exit, thereby completing the evacuation. 
 At all times, the two robots travel at unit speed. 
 Without loss of generality, we assume for our analysis that \RA~ finds the exit and then {\em catches}~\RB. 

 
 \paragraph{Meeting Protocol for $R_1$:}
If at any time $t_0$,  $\RA$  finds the exit at point $X$, it 
computes the shortest additional time $t$ such that $\RB$, after traveling 
distance $t_0+t$, is located at point $M$ satisfying $\barr{XM}=t$. 
Robot \RA\ moves along the segment $XM$. At time $t_0+t$ the two robots meet 
at $M$ and traverse  directly back to the exit at $X$ incurring total time cost 
$t_0+2t$.

\subsection{Evacuation Algorithm  \texorpdfstring{$\mathcal A$}{Lg} of \texorpdfstring{\cite{CGGKMP}}{Lg}}

We proceed by describing the trajectories of the two robots in Algorithm  
$\mathcal A$. As mentioned above, both robots start 
from the centre $O$ of the disk and move together to an arbitrary position $A$ on the 
boundary of the disk. $\RB$ then moves clockwise along the boundary of the disk up to distance $\pi$, see the
left-hand side of Figure~\ref{fig: old-algo}, and robot $\RA$ moves counter clockwise on the trajectory that is a reflection of $\RB$'s trajectory with respect to the 
line passing through $O$ and $A$. When $\RA$ finds the exit, 
it invokes the  {\it meeting protocol} in order to meet  $\RB$, after which the evacuation is completed. 

\begin{figure}[!ht]
                \centering
                \includegraphics[scale=0.5]{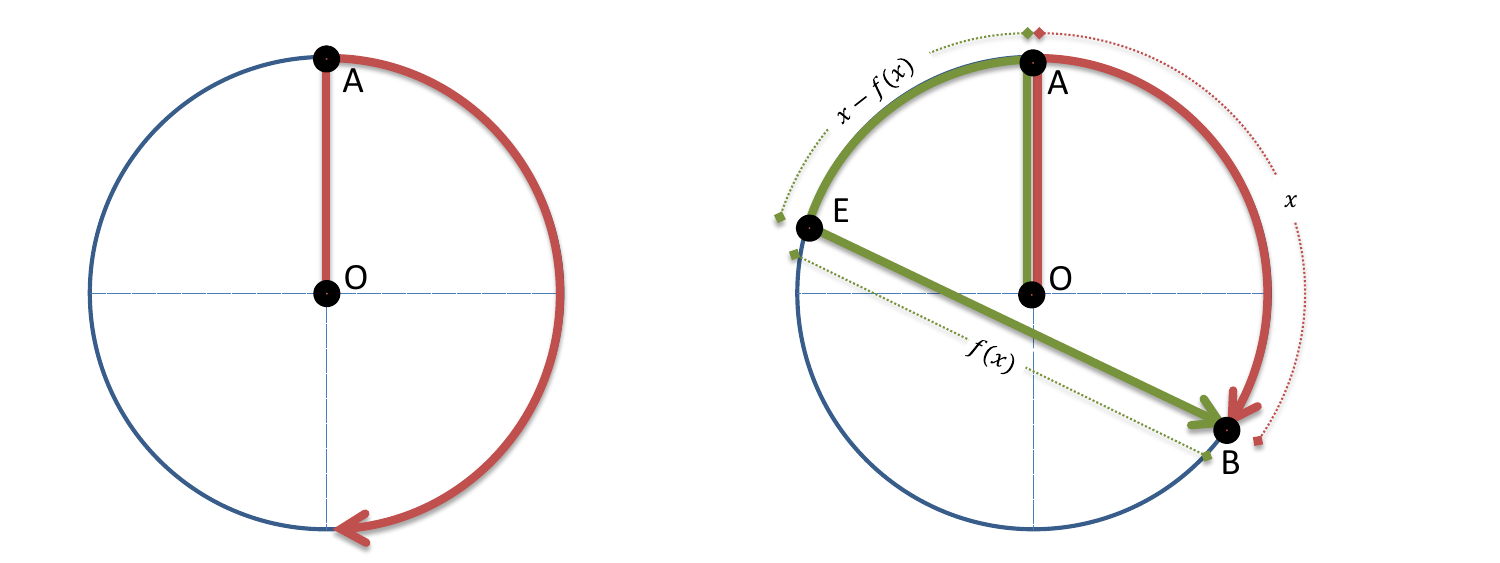}
				\caption{Evacuation Algorithm $\mathcal A$ with exit position $E$.  
				The trajectory of robot $\RB$ is depicted on the left. 
				The movement paths of robots $\RA,\RB$ are shown on the right, until the moment they meet at point $B$ on the circle. 
				}
                \label{fig: old-algo}
\end{figure}


The meeting-protocol trajectory of $\RA$ in Algorithm $\mathcal A$ is depicted in the right-hand side of Figure~\ref{fig: old-algo}, with exit point $E$ and meeting point $B$. 
Clearly, for the two robots to meet, we must have $\arccc{AB}=\arccc{EA}+\overline{EB}$. Next we want to analyze the performance of the algorithm, with respect to $x:=\arccc{AB}$, i.e. the length $x$ of the arc that $\RB$ travels, before it is met by $\RA$. We also set $f(x): = \overline{EB}$.
It follows that $\arccc{EA}=x-f(x)$, and therefore\footnote{We are using the elementary fact that in a unit circle, an arc of length $\alpha$ corresponds to a chord of length $2\sinn{\alpha/2}$}
\begin{equation}\label{equa: def f}
f(x)~=~ ~z,~ \textrm{ where }z \textrm{ is a solution of the equation } 
z= 2\sinn{x-\frac z2}.
\end{equation}
In other words,  $f(x)$\footnote{Uniqueness of the root of the equation defining $f(x)$ is an easy exercise.} 
is the length of the segment $EB$ that $\RA$ 
needs to travel in the interior of the disk after locating the exit at $E$, 
to meet $\RB$ at point $B$. 

Then the cost of Algorithm $\mathcal A$, given that the two robots meet at time $x$ after they together reached the boundary of the disk at $A$, is $1+x+f(x)$. Given that the distance $x-f(x)$ traveled by $\RA$ until finding the exit is between $0$ and $\pi$, it directly follows that $x$ can take any value between $0$ and $\pi$ as well. Hence, the worst case performance of Algorithm $\mathcal A$ is determined by
$$
\max_{x\in [0,\pi]} \{ x+f(x) \} 
$$
The next lemma, along with more details of its proof, follows from~\cite{CGGKMP}.
\begin{lemma}\label{lem: monotonicity of x+f(x)}
Expression $F(x):=x+f(x)$ attains its supremum at $x_0 \approx 2.85344$
(which is $ \approx 0.908279 \pi$). In particular, $F(x)$ is strictly increasing when $x \in [0, x_0]$ and strictly decreasing when $x \in [x_0, \pi]$.
\end{lemma}

\begin{proof}
The behavior of $F(x)$, as $x$ ranges in $[0,\pi]$, is shown in Figure~\ref{fig: performance old-algo}.
\begin{figure}[!ht]
                \centering
                \includegraphics[scale=0.3]{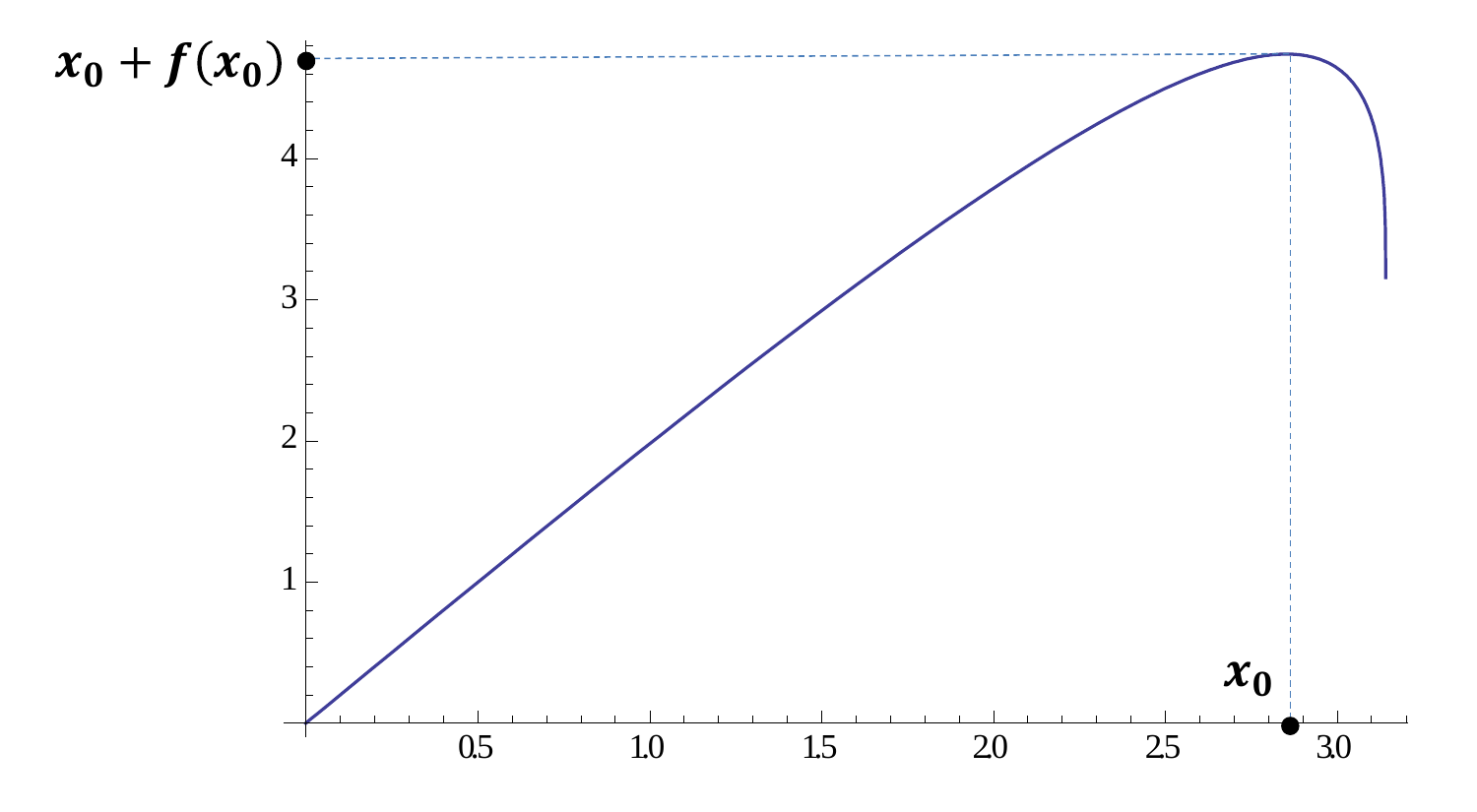}
                \caption{The performance of Algorithm $\mathcal A$ as a function of the meeting points of the robots.}
                \label{fig: performance old-algo}
\end{figure}
\end{proof}

By Lemma~\ref{lem: monotonicity of x+f(x)}, the evacuation time of Algorithm $\mathcal A$ is $1+x_0+f(x_0) < 5.740$. 
The worst case position of the exit is attained for $x_0-f(x_0) \approx 0.308 \pi$. 


\subsection{New evacuation algorithm \texorpdfstring{$\mathcal B(\chi, \phi)$}{Lg}}

We now show how to improve Algorithm $\mathcal A$ 
and obtain an evacuation time of at most~$5.644$.
The main idea for the improvement 
is to change the trajectory of the robots when the distance traveled on the boundary of the disk approaches the critical value $x_0$ of Lemma~\ref{lem: monotonicity of x+f(x)}. Informally, 
robot $\RB$ could meet $\RA$ earlier if it makes a linear detour inside the interior of 
the disk towards $\RA$ a little before traversing distance $x_0$. 

We describe a generic family of algorithms that realizes this idea. The specific
trajectory of each algorithm is determined by two parameters $\chi$ and $\phi$  where $\chi \in [\pi/2,x_0]$ and $\phi \in [0, f(\chi)/2 ]$, whose optimal values will be determined later. 
For ease of exposition, we assume that $\RA$ finds the exit. The trajectory of $\RB$ (assuming it has not yet met $\RA$) is  
partitioned into four phases that we call the \textit{deployment}, \textit{pre-detour}, \textit{detour} and \textit{post-detour} phase. 
The description of the phases relies on the left-hand side of Figure~\ref{fig: line-evac New}.\\

\noindent
{\bf Algorithm $\mathcal B (\chi, \phi)$}(with a linear detour).  
The phases of robot $\RB$'s trajectory (unless it is met by $\RA$ to go to the exit) are: \\ 
$\star$ \textit{Deployment phase:} 
Robot $\RB$ starts from the centre $O$ of the disk and moves to an arbitrary position $A$ on the boundary of the disk. \\ 
$\star$ \textit{Pre-detour phase:} 
$\RB$ moves clockwise along the boundary of the disk until having explored an arc of length $\chi$. 
Let $B$ be the point on the circle in which this phase ends.\\
$\star$ \textit{Detour phase:}
Let $D$ be the reflection of $B$ with respect to $AA'$ (where $A'$ is the antipodal point of $A$). 
$\RB$ moves on a straight line towards the interior of the disk and towards the side where $O$ lies, 
forming an angle of $\phi$ with line $BD$ until it reaches the line $AA'$ at point $C$. 
At $C$, $\RB$ turns around and follows the same straight line back to $B$. 
Note that by the restrictions on $\phi$, $C$ is indeed in the interior of the line segment $AA'$. \\
$\star$ \textit{Post-detour phase:}
Robot $\RB$ continues moving clockwise on the arc $BA'$. 
\\[0.5ex]
 The trajectory of $\RA$, until it finds the exit, is the reflection  
 of \RB's trajectory along the line $AA'$.
 When at time $t_0$, $\RA$ finds the exit, 
 it follows the Meeting Protocol defined above.\\


\begin{figure}[!ht]
                \centering
                \includegraphics[scale=0.5]{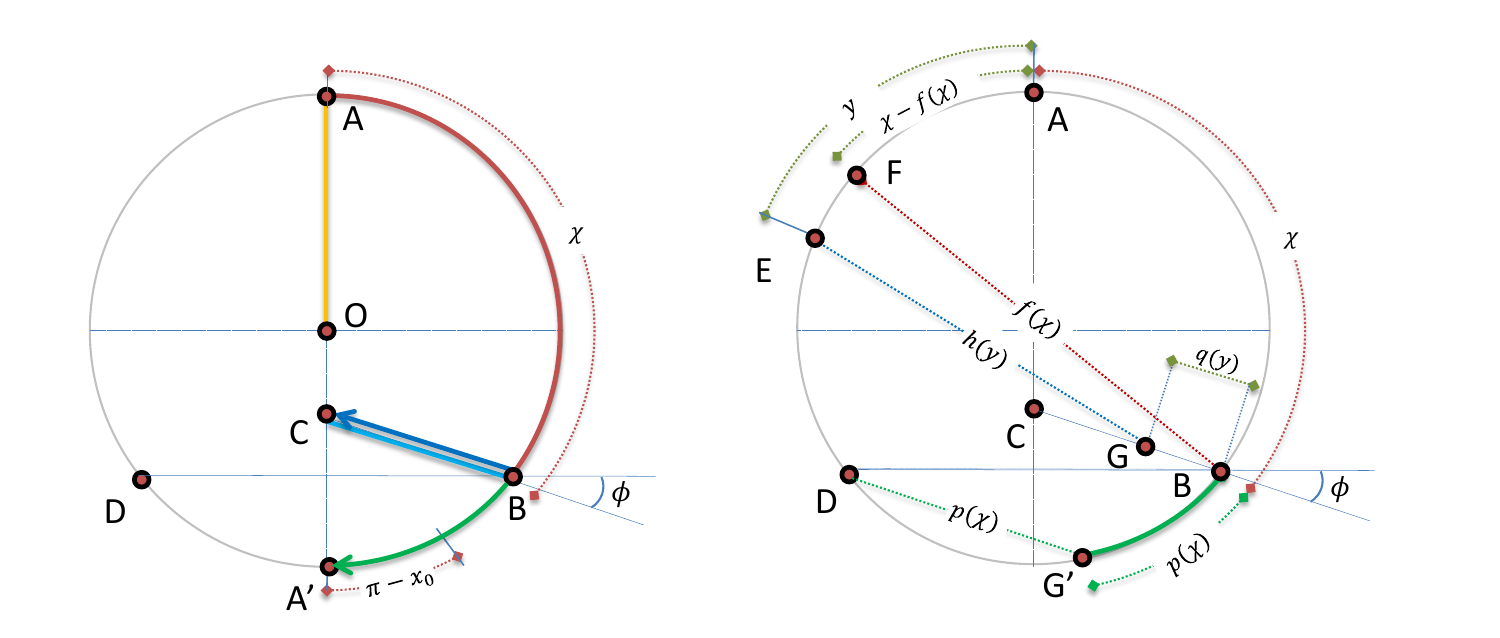}
                \caption{Illustrations for Algorithm $\mathcal B(\chi, \phi)$. }
                \label{fig: line-evac New}
\end{figure}

In what follows we assume that $\RA$ finds the exit on the arc $A'A$. Subsequently, it catches $\RB$ on its trajectory and they return together to the exit. 
There are three cases to consider as to where $\RB$ can be caught by \RA\ while moving on its trajectory. 
For all three cases, the reader can consult the right-hand side of Figure~\ref{fig: line-evac New}. As the time needed for the deployment phase is independent of where the exit is located, we ignore this extra cost of 1 during the
case distinction.

\begin{description}
\item{Case 1: $\RB$ is caught during its pre-detour phase:} The meeting point is anywhere on the arc $AB$. Recall that $\chi\leq x_0$, so by Lemma~\ref{lem: monotonicity of x+f(x)} the location $F$ of the exit on the arc $FA$ that maximizes the cost of $\mathcal B(\chi, \phi)$ is the one at 
at distance $\chi-f(\chi)$ from~$A$ (see right-hand side of Figure~\ref{fig: line-evac New}). The cost then is $\arccc{AB}+\barr{BF}=\arccc{FA}+2\barr{BF} =\chi + f(\chi)$. 

\item{Case 2: $\RB$ is caught during its detour phase:} Let $G$ be the point on $BC$ where the robots meet. Further, let $E$ be the position of the exit on the arc $A'A$, and let $y:=\arccc{EA}$. In the following, $h(y):=\barr{EG}$ denotes the length of the trajectory of $\RA$ in its attempt to catch $\RB$ after it finds the exit. Also, $q(y):=\barr{BG}$ denotes the distance that $\RB$ travels on $BC$ until it is caught by $\RA$. Note that the functions $h$ and $q$ also depend on $\chi$ and $\phi$; however, while those are fixed, $y$ varies with the position of the exit. 
	Lemma~\ref{lem: meeting condition for line algo} below states that $h(y)$ and $q(y)$ are well defined.
\begin{lemma}\label{lem: meeting condition for line algo}
	$y \in [\chi-f(\chi), \chi]$ if and only if the meeting point $G$ of the robots is on the line segment $BC$. Moreover, robot $\RB$ can be caught by $\RA$ only while moving from $B$ to $C$. 
\end{lemma} 
\begin{proof}
If the exit is located at point $F$ (i.e. $y=\chi-f(\chi)$), then the meeting point is $B$. We also observe that if the exit-position coincides with $D$ (i.e. $y=\chi$), then the meeting point is $C$. 
Recall that $\phi\leq f(\chi)/2 = \arccc{DF}/2$, and hence point $C$ is in the interior or on the boundary of the triangle $FDB$ as it is depicted in right-hand side of Figure~\ref{fig: line-evac New}. Therefore, after time $\chi$, robot $\RB$ would approach the exit if it was located anywhere on the arc ${DE}$. In particular, if $y \in [\chi-f(\chi),\chi]$ then the meeting point for Algorithm $\mathcal A$ would be on the 
arc $BA'$. In Algorithm $\mathcal B(\chi, \phi)$, $\RB$ has a trajectory that brings it closer to the exit. This guarantees that the meeting point $G$ always exists, and it lies in the 
line segment $BC$ (we will soon derive a closed formula relating $\arccc{EA}$ and $\barr{BG}$). 
The previous argument guarantees that $h(y)$ is continuous and strictly decreasing in $y$.  
Notice that $\arccc{EA}+\barr{EG}=\arccc{AB}+\barr{BG}$ (since the two robots start from the same position $A$), which means that 
\begin{equation}\label{equa: meeting condition line interior trajectory}
y+h(y)=\chi+q(y)
\end{equation}
 Hence, $q(y)= y+h(y)-\chi$, and as already explained, $q(\chi-f(\chi))=0$ and $q(\chi)=\barr{BC}$. By the mean value theorem, all values between 0 and $\barr{BC}$ are attainable for $q(y)$ and are attained while $y$ ranges in $[\chi-f(\chi), \chi]$. \hfill 
\end{proof}

We conclude that if the exit is located at point $E$, then the cost of the algorithm is $y+2h(y)$. Hence, in case 2, the cost of the algorithm is at most
$$
\max_{y \in [\chi - f(\chi) , \chi] } \{ y + 2 h(y) \}.
$$
Function $h(y)$ is calculated later in Lemma~\ref{lem: calculation of h, line interior}. 
We emphasize again that $h(y)$ and $q(y)$ also depend on the fixed parameters $\chi$ and $\phi$. 
\item{Case 3: $\RB$ is caught during its post-detour phase:} 
Due to the already searched domain of the pre-detour phase, in this case the exit lies in the interior of the arc $A'D$ or coincides with $A'$. 
At time $\postB=\chi+2q(\chi)$, robots $\RA$ and $\RB$ are located at points $D$ and $B$, respectively. Then they move towards each other on the arc $BD$ until $\RA$ finds the exit. 
Note that, since $\barr{DB}/2 = \sinn{\chi}$, 
we have $q(\chi)=\sinn{\chi}/\coss{\phi}$. 
By the monotonicity of the cost of Algorithm $\mathcal A$, see also Figure~\ref{fig: performance old-algo}, we have that the closer the exit is to $D$, the  higher is the cost of the evacuation algorithm. In the limit (as the position of the exit approaches $D$), the cost of case 3 approaches $\postB$ plus the time it takes $\RA$ to catch $\RB$ if the exit was located at $D$, and if they started moving from points $D$ and $B$ respectively. Let $G'$ be the meeting point on the arc $BD$ in this case, i.e. $\barr{DG'}=\arccc{BG'}$. 
We define $p(x)$ to be the distance that $\RA$ needs to travel in the interior of the disk to catch $\RB$, if the exit is located at distance $x$ from $A$. Clearly\footnote{Uniqueness of the root of the equation defining $p(x)$ is an easy exercise.} 
\begin{equation}\label{equa: def g}
p(x)~:=~ \textrm{unique}~z~\textrm{satisfying}~ z= 2\sinn{x+\frac z2}.
\end{equation}
Note also that $\barr{DG'} = p(\chi)$ so that the total cost in this case is at most  
$$\postB+2p(\chi) = \chi+2\sinn{\chi}/\coss{\phi} + 2p(\chi).$$
\end{description}

The following two lemmas summarize the above analysis and express $h(y)$ in explicit form (in dependence of $\chi$ and $\phi$), respectively.

\begin{lemma}\label{lem: line algo performance}
	For fixed constants $\chi$ and $\phi$ 
	the evacuation time of Algorithm $\mathcal B(\chi,\phi)$ is 
\begin{equation}\label{eqa: line algo performance}
1+\max
\left\{
\chi+f(\chi),~~
\max_{y \in [\chi - f(\chi) , \chi] } \{ y + 2 h(y) \},~~
\chi+2\sinn{\chi}/\coss{\phi} + 2p(\chi)
\right\},
\end{equation}
where $h(y)$ (that also depends on the choice of $\chi$ and $\phi$) denotes the time that a robot needs from the moment it finds the exit until it meets the other robot when following the meeting protocol.
\end{lemma}


\begin{lemma}\label{lem: calculation of h, line interior}
	For every $\chi>0$ and for every $\chi -f(\chi) \leq y \leq \chi$, the distance $h(y)$ that $\RA$ travels from $A$ until finding $\RB$ when following the meeting protocol in Algorithm $\mathcal B(\chi, \phi)$ is 
$$
h(y)=
\frac{
2+(\chi-y)^2 - 2\cos(\chi+y) + 2 (\chi-y) \left(\sin (\phi+y)- \sin (\phi-\chi) \right)
}{2 (\chi-y-\sin (\phi-\chi)+\sin (\phi+y))}.
$$
In particular, $h(y)$ is strictly decreasing for $0\leq \phi \leq f(\chi)/2$. 
\end{lemma}

\begin{proof}
\begin{figure}[thb]
                \centering
                \includegraphics[scale=0.5]{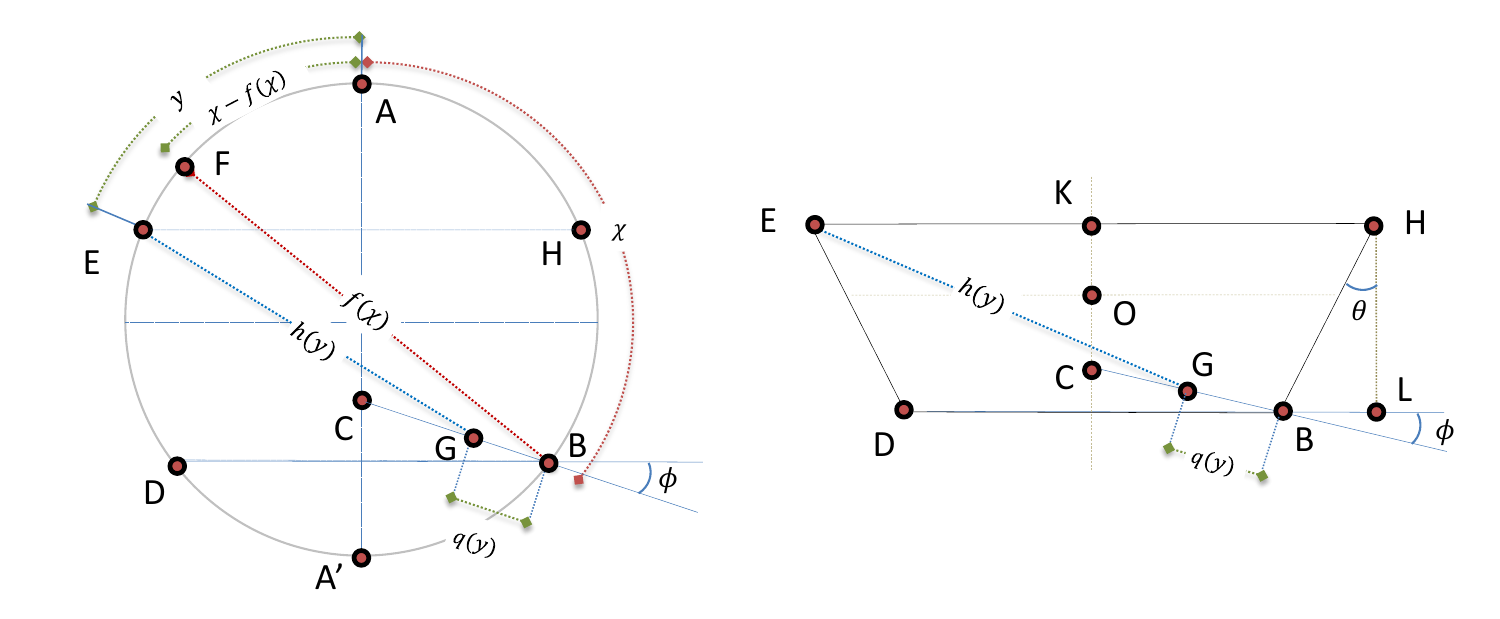}
                \caption{The analysis of Algorithm $\mathcal B(\chi, \phi)$. }
                \label{fig: line-evac-analysisAngled}
\end{figure}
We start by making some handy observations. 
For this we rely on Figure~\ref{fig: line-evac-analysisAngled} (that is a continuation of Figure~\ref{fig: line-evac New}). 
Let $H$ be the point that is symmetric to $E$ with respect to~$AA'$. 
Denote with $L$ the projection of $H$ onto the supporting line of $DB$. 
Set $\theta:=\angle{BHL}$, and observe the following equation for $\theta$.
\begin{align}
\theta &= \notag
\frac \pi 2 - \angle{EHB}
= \frac \pi 2 - \frac{\arccc{BE}}{2} 
= \frac \pi 2 - \frac{\arccc{BD}+\arccc{DE}}{2}\\
&= \label{equa: value of theta}
\frac \pi 2 - \frac{2(\pi-\chi)+\chi-y}{2}
=\frac{\chi+y}{2}-\frac{\pi}{2}
\end{align}
Our goal is to compute $h(y)=\barr{EG}$. For this we see that
$
\vect{EG}=\vect{EH}+\vect{HB}+\vect{BG},
$
and therefore
\begin{equation}\label{equa: inner products line algo}
\barr{EG}^2=\barr{EH}^2+\barr{HB}^2+\barr{BG}^2 
+ 2\left(
\vect{EH}\cdot\vect{HB}+
\vect{EH}\cdot\vect{BG}+
\vect{HB}\cdot\vect{BG}
\right).
\end{equation}
We have
$\barr{EH}=2\sinn{y}, \barr{HB}= 2\sinn{\frac{\chi-y}2}, \barr{BG}=q(y)$, 
and
\begin{align*}
\vect{EH}\cdot\vect{HB} &= 
\barr{EH} ~\barr{HB} \coss{\pi - \angle{EHB}}
\\
\vect{EH}\cdot\vect{BG} &=
\barr{EH} ~\barr{BG} \coss{\angle{GBL}}
\\
\vect{HB}\cdot\vect{BG} &=
\barr{HB} ~\barr{BG} \coss{\pi- \angle{GBH}}.
\end{align*}
We also have
\begin{align*}
\coss{\pi - \angle{EHB}}&=\coss{\pi/2+\theta}= \coss{\frac{\chi+y}2}
\\
\coss{\angle{GBL}} &=\coss{\pi-\phi}=\coss{\phi}
\\
\coss{\pi- \angle{GBH}} &= \coss{\pi +\frac\pi2+\theta - \phi}=
-\coss{\frac{\chi+y}2 - \phi}.
\end{align*}
Substituting the above in~\eqref{equa: inner products line algo}, we obtain an equation between $h(y)$, $q(y)$, $y$, $\chi$, and $\phi$. 
In the latter equation we can substitute $q(y)$ using the meeting condition~\eqref{equa: meeting condition line interior trajectory}, obtaining this way the required closed formula for $h(y)$.
The fact that $h(y)$ is strictly decreasing for $0\leq \phi \leq f(\chi)/2$ follows from the proof of Lemma~\ref{lem: meeting condition for line algo}.
\end{proof}

The first natural attempt in order to beat Algorithm $\mathcal A$ would be to consider $\mathcal B(\chi, 0)$, i.e. make $BC$ perpendicular to $AA'$ in Figure~\ref{fig: line-evac New}. In light of Lemma~\ref{lem: calculation of h, line interior} and using Lemma~\ref{lem: line algo performance}, 
we state the following claim to build some intuition for our next, improved, algorithm. 

\begin{claim}\label{claim: phi=0, line algo performance}
The performance of algorithm $B(\chi, 0)$ is optimized when $\chi=\chi_0\approx2.62359$, and its cost is $1+4.644=5.644$.
The location of the exit inducing the worst case for $B(\chi, 0)$ 
is when $y= \arccc{EA} \approx 0.837 \pi$. 
The meeting point of the two robots takes place at point $G$ (see Figure~\ref{fig: line-evac New}, and set $\phi=0$), where $q(y)= \barr{BG} \approx 0.117 \approx 0.236 \barr {BC}$. 
In particular, the cost of the algorithm, if the meeting point of the robots is during \RB's \textit{pre-detour}, \textit{detour} and \textit{post-detour} phase, is (approximately) 5.621, 5.644 and 5.644 respectively. 
\end{claim}

\ignore{
 (*Mathematica code*)
xx = 2.62359;
f[x_] := z /. FindRoot[ z == 2*Sin[x - z/2], {z, 1}];
Hp[x_, y_, 
   phi_] := (x^2 - 2 x y + y^2 + 4 Sin[(x - y)/2]^2 + 
     4 (x - y) Cos[phi] Sin[y] + 4 Sin[x] Sin[y] + 
     4 (x - y) Sin[(x - y)/2] Sin[
       phi + ArcSin[Cos[(x + y)/2]]])/(2 (x - y + 2 Cos[phi] Sin[y] + 
       2 Sin[(x - y)/2] Sin[phi + ArcSin[Cos[(x + y)/2]]]));

NMaximize[ {y + 2*Hp[xx, y, 0], y >= xx - f[xx], y <= xx}, y]

yy = 0.8366130642397082;
yy + Hp[xx, yy, 0] - xx  (*this is q*)
}

Note that $\chi_0$ of Claim~\ref{claim: phi=0, line algo performance} is strictly smaller than $x_0$ of Lemma~\ref{lem: monotonicity of x+f(x)}. In other words, the previous claim is in coordination with our intuition that if the robots moved towards the interior of the disk a little before the critical position of the meeting point $x_0$ of Algorithm $\mathcal A$, then the cost of the algorithm could be improved. 

\subsection{New evacuation algorithm \texorpdfstring{$\mathcal C(\chi, \phi, \lambda)$}{Lg}}


Claim~\ref{claim: phi=0, line algo performance} is instructive for the following reason. 
Note that the worst meeting point~$G$ for Algorithm $\mathcal B(\chi_0,0)$ satisfies $\barr{BG} \approx 0.236 \barr {BC}$. 
This suggests that if we consider algorithm $\mathcal B(\chi_0,\phi)$ instead, where $\phi>0$, then we would be able to improve the cost if the meeting point happened during the detour phase of $\RB$. On one hand, this further suggests that we can decrease the detour position $\chi_0$ (note that the increasing in $\chi$ cost $\chi + f(\chi)$ is always a lower bound to the performance of our algorithms when $\chi<x_0$). On the other hand, that would have a greater impact on the cost when the meeting point is in the post-detour phase of $\RB$, as in this case the cost of moving from $B$ to~$C$ and back to $B$ would be $2\sinn{\chi}/\coss{\phi}$ instead of just  $2\sinn{\chi}$. 
A compromise to this would be to follow the linear detour trajectory of $\RB$ in $\mathcal B(\chi_0,\phi)$ only up to a certain threshold-distance $\lambda$, after which the robot should reach the diameter segment
$AA'$ along a linear segment perpendicular to  segment $AA'$
then return to the detour point $B$ along a linear segment.
Thus the detour forms a triangle.
This in fact completes the high level description of Algorithm
$\mathcal C(\chi, \phi, \lambda)$ that we formally describe below. 

For an instance of Algorithm $\mathcal C(\chi, \phi, \lambda)$,
we fix $\chi$, $\phi$, and $\lambda$, with $\chi \in [\pi/2,x_0]$, $\phi \in [0, f(\chi)/2 ]$, and $\lambda \in [0, \sinn{\chi}/\coss{\phi}]$.
%
As before, we assume without loss of generality that $\RA$ finds the exit. 
%
The trajectory of robot $\RB$ (that has neither found the exit nor met \RA\ after \RA\ has found the exit) can be partitioned into roughly the same four phases as for Algorithm $\mathcal B(\chi, \phi)$; so we again call them \textit{deployment}, \textit{pre-detour}, \textit{detour} and \textit{post-detour} phases. The description of the phases refers to the left-hand side of Figure~\ref{fig: piecewise-line-evac New}, which is a partial modification of Figure~\ref{fig: line-evac New}.\\

\noindent
{\bf Algorithm $\mathcal C (\chi, \phi, \lambda)$} (with a triangular detour).  
The phases of robot $\RB$'s trajectory (unless it is met by $\RA$ to go to the exit) are: \\ 
$\star$ \textit{Deployment phase:} 
Same as in Algorithm $\mathcal B(\chi, \phi)$. At time 1, $\RB$ is at point~$A$. 
\\
$\star$ \textit{Pre-detour phase:} 
Same as in Algorithm $\mathcal B(\chi, \phi)$. In additional time $\chi$, $\RB$ is in point~$B$.
\\
$\star$ \textit{Detour phase:} This phase is further split into three subphases. \\
\hspace*{.5cm}$\diamond$ \textit{Subphase-1:} Up to additional time $\lambda$, $\RB$ moves along a line segment exactly as in the detour phase of Algorithm $\mathcal B(\chi, \phi)$. Let $G$ be the position of the robot at the end of this phase. 
\\
\hspace*{.5cm}$\diamond$ \textit{Subphase-2:} Let $C$ be the projection of $G$ onto $AA'$. $\RB$ follows line segment $GC$ until it reaches point $C$. 
\\
\hspace*{.5cm}$\diamond$ \textit{Subphase-3 (Recovering phase):} Robot follows
line segment $CB$ back to point $B$.
\\
$\star$ \textit{Post-detour phase:}
Same as in Algorithm $\mathcal B(\chi, \phi)$. After additional time $\arccc{EA'}$, $\RB$ reaches point $A'$. \\
 At the same time \RA~ follows a trajectory that is the reflection of \RB's
 trajectory along the line $AA'$ until it finds the exit. 
 The meeting protocol that $\RA$ follows once it finds the exit is the same as for Algorithms $\mathcal A$ and $\mathcal B(\chi, \phi)$. 

\begin{figure}[!ht]
                \centering
                \includegraphics[scale=0.5]{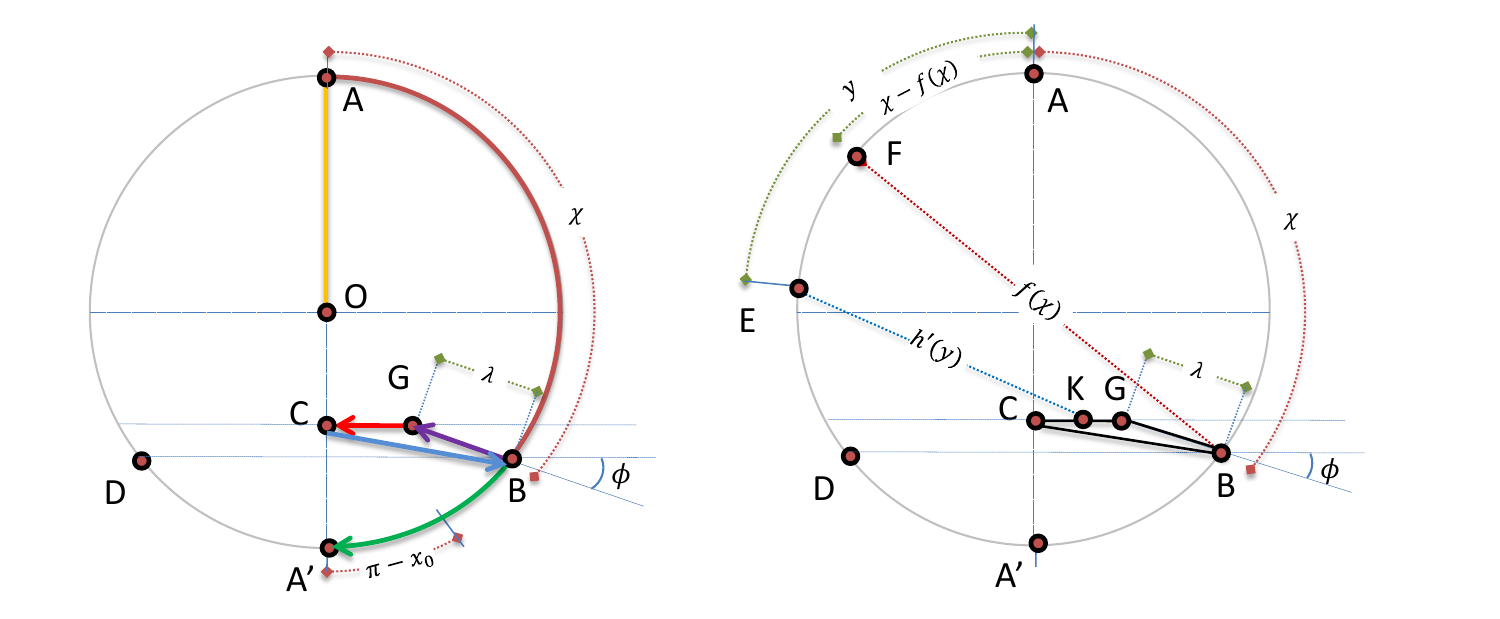}
                \caption{Illustrations for Algorithm $\mathcal C (\chi, \phi, \lambda)$.}
                \label{fig: piecewise-line-evac New}
\end{figure}

Obviously, Algorithm $\mathcal C\left(\chi, \phi, \frac{\sinn{\chi}}{\coss{\phi}}\right)$ is identical to Algorithm $\mathcal B(\chi, \phi)$.
Notice that an immediate consequence of the definition is that if robot $\RA$ finds the exit and meets $\RB$ during its detour subphase-2 in some point $K$ (as in the right-hand side of Figure~\ref{fig: piecewise-line-evac New}), then 
\begin{equation}\label{equa: meeting condition piecewise line}
\arccc{EA}+\barr{EK}=\arccc{AB}+\barr{BG}+\barr{GK}.
\end{equation}

When $\RA$ finds the exit somewhere on the arc $A'A$,
it catches $\RB$ on its trajectory so that they return together to the exit. Note that since robots meet at point $C$, if the exit is not in the arc $DB$, it is impossible for a robot to be caught by the other robot in subphase-3 of its detour phase. Hence, there are four cases as to where $\RB$ can be caught by $\RA$ that found the exit. 
As before, we omit the extra cost 1 which is the time needed for the deployment phase during the case distinction.\\

\begin{description}
\item{Case 1: $\RB$ is caught in its pre-detour phase:} The cost of Algorithm $\mathcal C(\chi, \phi, \lambda)$ is at most $\chi+f(\chi)$, exactly as in the analogous case of Algorithm $\mathcal B(\chi, \phi)$. 

\item{Case 2:  $\RB$ is caught in its detour subphase-1:} As in case 2 of the analysis of Algorithm $\mathcal B(\chi, \phi)$, if $E$ is the position of the exit, then $y=\arccc{EA}$ satisfies $y\geq \chi-f(\chi)$ (the relevant figure in this case remains the left-hand side of Figure~\ref{fig: line-evac-analysisAngled}). As long as $q(y)$ remains less than $\lambda$, the cost of the algorithm remains $y+2h(y)$. In order to find the maximum $y$ for which this formula remains valid, we need to solve the equation $\lambda=q(y)$. This is possible by recalling that  $h(y)=\chi+q(y)-y$, and by invoking the formula of $h(y)$ as it appears in Lemma~\ref{lem: calculation of h, line interior}. By the monotonicity of $h(y)$, we have that there exists a unique $\psi$ satisfying $h(\psi) = \chi+\lambda - \psi$. It follows that the cost of the algorithm in this case is at most
$\max_{\chi-f(\chi) \leq y \leq \psi} \{ y + 2h(y)\}$.

\item{Case 3: $\RB$ is caught in its detour subphase-2:} In this case, the relevant figure is the right-hand side of Figure~\ref{fig: piecewise-line-evac New}. Let the exit be at point $E$, and let $K$ denote the meeting point of the robots on the line segment $GC$. We set $h'(y):=\barr{EK}$, 
whose value is dertermined by 	
Lemma~\ref{lem: quantities from algo C} (a).  We conclude that in this case the cost of the algorithm is at most $\max_{\psi \leq y \leq \chi} \{ y + 2h'(y)\}$.

\item{Case 4:  $\RB$ is caught in its post-detour phase:} Let $\postC$ again be the total time a robot needs until it enters its post-detour phase. As in case 3 of Algorithm $\mathcal B(\chi, \phi)$, the cost of Algorithm $C(\chi \phi, \lambda)$ for this case is at most $\postC+2p(\chi)$. It thus remains to show how to calculate $\postC = \arccc{AB}+\barr{BG}+\barr{GC}+\barr{CB}$, which is done in Lemma~\ref{lem: quantities from algo C}~(b). 
\end{description}

\begin{lemma}\label{lem: quantities from algo C}
The following statements hold for Algorithm $\mathcal C(\chi, \phi, \lambda)$: \\
(a) Suppose that $\RA$ finds the exit and meets $\RB$ in its detour subphase-2. Then the time $h'(y)$ that 
$\RA$ needs from finding the exit until meeting $\RB$ is $h'(y) = N(\chi, y, \lambda, \phi) / D(\chi, y, \lambda, \phi)$, where
\begin{align*}
N(\chi, y, \lambda, \phi):=
&
~~~~2+\lambda^2+(\lambda+\chi-y)^2+2 \lambda (\sinn{\phi-\chi}-\sinn{\phi+y}) \\
&+2 (\lambda+\chi-y) (\sinn{\chi}+\sinn{y} 
	- \lambda \coss{\phi})-2 \coss{\chi+y}, \ \ \mbox{ and} \\
D(\chi, y, \lambda, \phi):=
&2 (\lambda+\chi - y +\sinn{\chi}+\sinn{y} -\lambda \coss{\phi}).
\end{align*}
(b) Suppose that $\RA$ finds the exit and meets $\RB$ in its post-detour phase. 
Then the total time that $\RB$ spends in its detour phase is
$$
\lambda+\sinn{\chi}-\lambda\coss{\phi}+\sqrt{\sin^2(\chi)+\lambda^2\sin^2(\phi)}.
$$
\end{lemma}

\begin{proof}
As an illustration of the proof we refer to Figure~\ref{fig: piecewise-line-evac-analysisAngled}, which is a continuation of Figure~\ref{fig: piecewise-line-evac New}. Let $H$ be the symmetric point of $E$ with respect to diameter $AA'$. As in Figure~\ref{fig: line-evac-analysisAngled}, $L$ is the projection of $H$ onto the supporting line of $DB$, and $\theta$ denotes the angle $\angle{BHL}$, whose value is given by \eqref{equa: value of theta}. Further, $G'$ and $C'$ are the projections of $G$ and $C$, respectively, onto $DL$. The calculations below follow the spirit of the arguments in Lemma~\ref{lem: calculation of h, line interior}. 
\begin{figure}[!ht]
                \centering
                \includegraphics[scale=0.5]{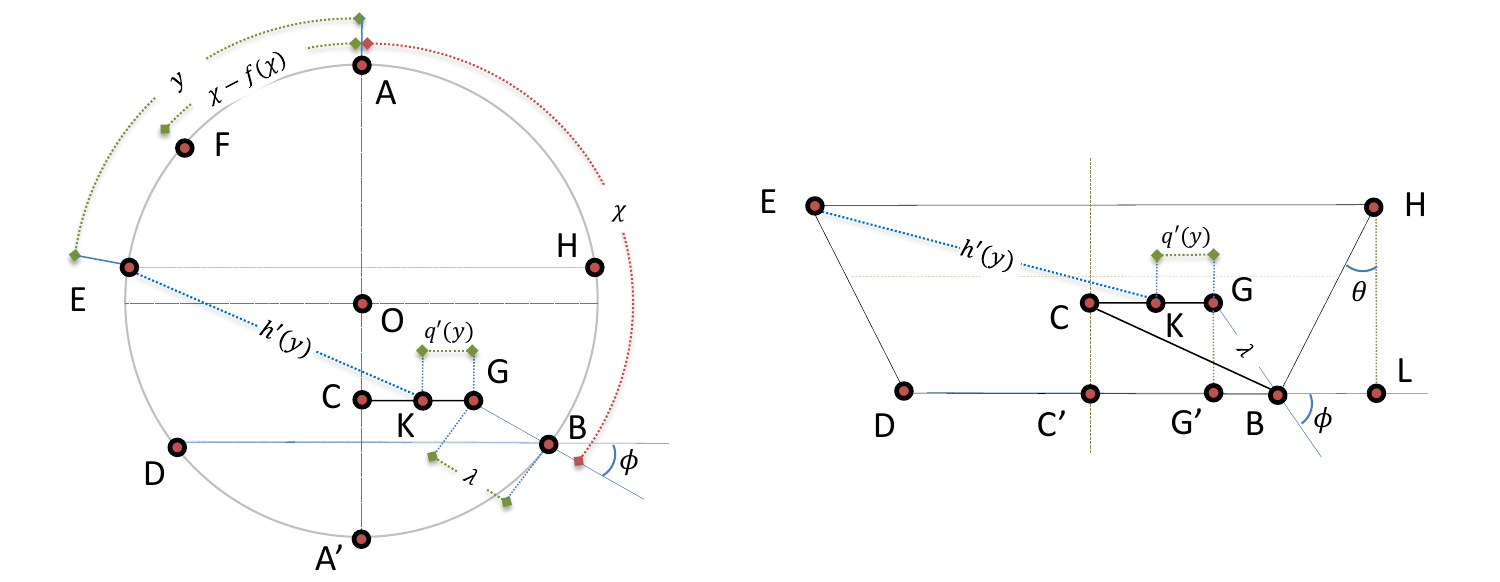}
                \caption{The analysis of Algorithm $\mathcal C (\chi, \phi, \lambda)$.}
                \label{fig: piecewise-line-evac-analysisAngled}
\end{figure}

(a) As before, $y$ denotes the distance of the exit from point $A$. 
We have that 
$
\vect{EK}=\vect{EH}+\vect{HB}+\vect{BG}+\vect{GC},
$
and therefore
\begin{align}
\barr{EK}^2
= & ~~ \barr{EH}^2+\barr{HB}^2+\barr{BG}^2  + \barr{GC}^2 +  \notag \\
&
 2\left(
\vect{EH}\cdot\vect{HB}+
\vect{EH}\cdot\vect{BG}+
\vect{EH}\cdot\vect{GC}+
\vect{HB}\cdot\vect{BG}+
\vect{HB}\cdot\vect{GC}+
\vect{BG}\cdot\vect{GC}
\right), 
\label{equa: inner products piecewise line algo}
\end{align}
where
$\barr{EK}=h'(y)$, $\barr{EH}=2\sinn{y}$, $\barr{HB}= 2\sinn{\frac{\chi-y}2}$, $\barr{BG}=\lambda$, and $\barr{GK}=q'(y)$. 
The inner products $\vect{EH}\cdot\vect{HB}, \vect{EH}\cdot\vect{BG}, \vect{HB}\cdot\vect{BG}$ are obtained exactly as in Lemma~\ref{lem: calculation of h, line interior}. For the remaining inner products we see that 
\begin{align*}
\vect{EH}\cdot\vect{GK} &= 
\barr{EH} ~\barr{GK} \coss{\pi}
=- \barr{EH} ~\barr{GK},
\\
\vect{HB}\cdot\vect{GK} &=
\barr{HB} ~\barr{GK} \coss{\frac\pi2-\theta}
=\barr{HB} ~\barr{GK} \sinn{\theta}, \ \ \mbox{ and}
\\
\vect{BG}\cdot\vect{GK} &=
\barr{BG} ~\barr{GK} \coss{\phi}.
\end{align*}
Substituting the above in~\eqref{equa: inner products piecewise line algo}, we obtain an equation for $h'(y)$ as a function of
$q'(y)$, $y$, $\chi$, and $\phi$. In that equation we can substitute $q'(y)$ using the meeting condition~\eqref{equa: meeting condition piecewise line}, according to which $q'(y)=y+h'(y)-\chi-\lambda$. 
Resolving the resulting equation for $h'(y)$ then gives the desired formula. \\

(b) We need to determine $\barr{BG}+\barr{GC}+\barr{CB}$, where $\barr{BG}=\lambda$. First, we observe that 
$
\barr{GC} = \barr{G'C'}=\barr{BC'}-\barr{BG'} = \sinn{\chi} - \lambda \coss{\phi}.
$
In order to calculate $\barr{CB}$, we see that $\vect{BC}=\vect{BG}+\vect{GC}$. Hence we obtain
\begin{align*}
\barr{BC}^2
&=\barr{BG}^2+\barr{GC}^2 + 2\vect{BG}\cdot\vect{GC} \\
&= \lambda^2 + \left(\sinn{\chi} - \lambda \coss{\phi}\right)^2 + 2\lambda \left(\sinn{\chi} - \lambda \coss{\phi}\right) \coss{\phi} \\
& = \sin^2(\chi)+\lambda^2\sin^2(\phi),
\end{align*}
which concludes our claim.  
\end{proof}

Before stating our main theorem, we summarize the total time required by Algorithm $\mathcal C(\chi, \phi, \lambda)$ in the following lemma.

\begin{lemma}\label{lem: cost of piecewise line algo with parameters}
The cost of Algorithm $\mathcal C(\chi, \phi, \lambda)$ can be expressed as 
$$
1+\max
\left\{
\begin{array}{lr}
\chi + f(\chi) & \textit{(pre-detour phase)} \\
\max_{\chi-f(\chi) \leq y \leq \psi} \{ y + 2h(y)\}& \textit{(detour subphase-1)} \\
\max_{\psi \leq y \leq \chi} \{ y + 2h'(y)\} & \textit{(detour subphase-2)} \\
\chi + \lambda+\sinn{\chi}-\lambda\coss{\phi}+
&\sqrt{\sin^2(\chi)+\lambda^2\sin^2(\phi)} + 2p(\chi)\\
  &\textit{(post-detour phase)}
\end{array}
\right\},
$$
where the functions $f$ and $p$ are as in~\eqref{equa: def f} and \eqref{equa: def g}, respectively; functions $h(y)$ and $h'(y)$ are expressed explicitly in Lemmas~\ref{lem: calculation of h, line interior} and \ref{lem: quantities from algo C}~(a), respectively; and $\psi$ is the unique solution to the equation $h(\psi) = \chi+\lambda - \psi$. 
\end{lemma}

In Lemma~\ref{lem: cost of piecewise line algo with parameters}, the phases in parentheses in the maximization expression are the ones in which $\RA$ meets $\RB$ after having found the exit.
Using the statement of Lemma~\ref{lem: cost of piecewise line algo with parameters} and numerical optimization, we obtain the following improved upper bound.

\begin{theorem}\label{thm: num performance piecewise line algo}
For $\chi_0=2.631865, \phi_0=0.44916$ and $\lambda_0=0.05762$, the evacuation algorithm $\mathcal C(\chi_0, \phi_0, \lambda_0)$ has cost no more than $5.628$. 
\end{theorem}

\begin{proof}
We examine the cost of our algorithm depending on where the meeting point of the two robots occurs. The guidelines of the analysis are suggested by Lemma~\ref{lem: cost of piecewise line algo with parameters}. Also the deployment cost of 1 will be added at the end. Any calculations below are numerical, and were performed using \textsc{mathematica}.

For the given parameters, we see that $f(\chi_0) = 1.99603, p(\chi_0)=0.506932, \chi_0-f(\chi_0)=0.63584$, and $\psi=0.755204$. If the meeting point is during the pre-detour phase, then the cost is $\chi_0 + f(\chi_0) <4.62791$ (note that $\chi_0<x_0$). If the meeting point is in the post-detour phase, then the cost is $\chi_0 + \lambda_0+\sinn{\chi_0}-\lambda\coss{\phi_0}+\sqrt{\sin^2(\chi_0)+\lambda_0^2\sin^2(\phi_0)} + 2p(\chi_0) < 4.627965$. 

For the more interesting intermediate cases, we have 
\begin{align*}
h(y) &= 
\frac{
-0.5 y^2
+3.45042 y
+ a(y) 
+ b(y) 
-6.61768
}{
y-0.900812 \sin (y)-0.434209 \cos (y)-3.45042
}
\\
h'(y)&=
\frac{
-0.5 y^2+3.12552 y
+(y-3.12552) \sin (y)
- 0.847858 \cos(y)
-5.74387
}{
y
-\sin (y)
-3.12552},
\end{align*}
where $a(y) := (0.900812 y-2.85875) \sin (y)$
and $b(y) := (0.434209 y-2.01566) \cos (y)$.
Hence the cost in the detour subphase-2 is bounded from above by 
$$\max_{0.63584 \leq y \leq 0.755204} \{ y + 2h(y)\} < 4.627972$$ and the cost in the detour subphase-3 is bounded from above by 
$$\max_{0.755204 \leq y \leq 2.631865} \{ y + 2h'(y)\} < 4.627961.$$
This completes the proof of the theorem. 
\end{proof}

\section{Lower Bound}
\label{sec:Lower Bound}

In this section we show that any evacuation algorithm for
two robots in the \ff model takes time at least
$3+\frac{\pi}{6}  + \sqrt{3} \approx 5.255$. 
We first prove a result of independent interest about an evacuation problem on a hexagon. 
\begin{theorem}\label{thm:hexagon}
Consider a hexagon of radius $1$ with an exit placed at an unknown vertex.
Then the worst case evacuation time for two robots starting at any two  arbitrary vertices of the hexagon 
is at least $2 + \sqrt{3}$.
\end{theorem}
\begin{proof}
Assume an arbitrary deterministic algorithm $\mathcal D$ for the problem. 
$\mathcal D$ solves the problem for any input, i.e., any placement of the exit.
We construct two inputs for $\mathcal D$ and show that for at least one of them, the required evacuation time is at least $2 + \sqrt{3}$.
First, we let $\mathcal D$ run without placing an exit at any vertex, so as to 
find out in which order the robots are exploring all the vertices
of the hexagon. We label the vertices of the hexagon according to this order (if two vertices are explored simultaneously then we just order them arbitrarily).
	Let $t$ be the time when the fifth vertex, $v_5$, of the hexagon is visited by some robot, say $\RA$, i.e., robot $\RA$ is at vertex $v_5$ at time $t$, and four more vertices of the hexagon have already been visited. 
In other words, $v_5$ and $v_6$ are the only vertices\footnote{It might be that $v_4$ and $v_5$ are explored simultaneously, or that $v_5$ and $v_6$ are explored simultaneously. In the former case $v_6$ is explored strictly after $v_5$ while in the latter $v_4$ is explored strictly before $v_5$.} 
that are guaranteed to not have been explored at time $t-\varepsilon$, for any sufficiently small $\varepsilon > 0$.
Note that we must have $t\geq2$, since at least one of the two robots must have visited at least three vertices by time $t$ (and hence must have walked at least the 
two segments between the first and the second, and between the second and the third vertex visited in its trajectory).
	
The first input $I_1$ we construct has the exit placed at vertex $v_6$. 
Until robot $\RA$ reaches $v_5$, the algorithm $\mathcal D$ 
processes $I_1$ identically to the case in which there is no exit; further, at time $t$, robot $\RA$ needs additional time at least $1$ just to reach the exit.
If $t \geq 1 + \sqrt{3}$, then this input gives an evacuation time of at least $2 + \sqrt{3}$. 

\begin{figure}[ht]
                \centering
                \includegraphics[scale=0.7]{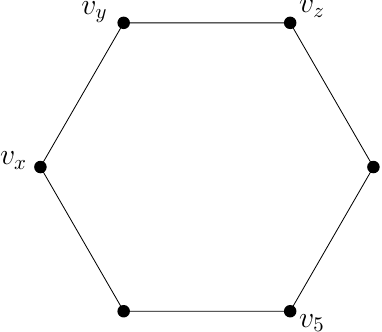}
                \caption{Vertices of the hexagon as visited by algorithm $\mathcal D$; $t$ is the time when the fifth vertex $v_5$ is visited by some robot, say $\RA$. One of the vertices adjacent to $v_5$ has not been visited yet by a robot (here $v_6 \not\in \{v_x,v_y,v_z\}$).}
                \label{fig: hex}
\end{figure}
Hence 
assume that $2 \leq t < 1 + \sqrt{3}$. Let $v_x$, $v_y$, and $v_z$ be the three vertices that are non-adjacent to $v_5$ in the hexagon (see Figure~\ref{fig: hex}).
Note that the minimum distance between $v_5$ and any of $v_x$, $v_y$, and $v_z$ is at least $\sqrt{3}$.
If $v_6 \in \{v_x,v_y,v_z\}$ then on input $I_1$, $\mathcal D$ needs evacuation time at least $t+\sqrt{3} \geq 2+\sqrt{3}$, as $\RA$ still has to 
reach the exit. 

Therefore, 
assume that $v_6 \not\in \{v_x,v_y,v_z\}$ and hence 	$\{v_x,v_y,v_z\} \subset \{v_1,v_2,v_2,v_4\}$.
Note that, since $t < 1+ \sqrt{3}$, on input $I_1$, robot
$\RA$ has visited at most one of $v_x$, $v_y$, and $v_z$ at time $t$. 
Hence, the other robot, $\RB$ has visited at least two of them.
For the second input $I_2$ that we construct, we place the exit on the vertex $v^*$ that is the last vertex among $v_x$, $v_y$, and $v_z$ in the visiting order of the vertices by $\RB$.
Let $t^*$ be the time when $\RB$ reaches $v^*$, and note that at least until $t^*$, the algorithm $\mathcal D$ behaves identical on the two inputs $I_1$ and $I_2$.
As $\RB$ has visited at least one vertex before visiting $v^*$, we have $t^* \geq 1$.

Next we claim that $\RA$ and $\RB$ cannot meet between time $t^*$ and $t$.
The claim below states that this is impossible. Namely, we can prove:
\begin{claim}\label{claim:hex}
Let $t$, $t^*$, and $I_2$ be as defined in the preceding paragraphs. 
If $t < 1+ \sqrt{3}$, then on input $I_2$, $\RA$ and $\RB$ do not meet between time $t^*$ 
and time $t$.
\end{claim}
\begin{proof} (Claim~\ref{claim:hex})
Assume on the contrary, that
on input $I_2$, $\RA$ and $\RB$ do meet at some time $t'$ at point $P$, with
$t^* \leq t' < t$.
Observe that on input $I_2$, robot $\RA$ continues until time $t'$ as on input $I_1$ but having met $\RB$ at time $t'$ might continue differently after time~$t'$.
Let $t_B=t'-t^*$ be the time that $\RB$ uses on input $I_2$ to get from the exit $v^*$ to $P$,  
and let $t_A = t-t'$ be the time that $\RA$ uses on input $I_1$ to get to vertex $v_5$ from $P$.  As $v^*$ and~$v_5$ are at distance at least $\sqrt{3}$, and since $t^* \geq 1$, we have
$\sqrt{3} \leq t_A + t_B = t-t' + t'-t^* = t-t^* \leq t-1$.  So we obtain $\sqrt{3}+1 \leq t$, which contradicts the assumption that $t < 1+\sqrt{3}$.
This proves the claim. 
\end{proof}

Having proved the claim, we conclude that on input $I_2$, 
$\RA$ continues until time $t$ as on input $I_1$. Hence $\RA$ needs at least $t+\sqrt{3} \geq 2+\sqrt{3}$ time to reach the exit on input $I_2$,
which completes the proof of the theorem.
\end{proof}

It is worth noting that the lower bound from Theorem~\ref{thm:hexagon} matches the upper bound of evacuating a regular hexagon, when the initial starting vertices may be chosen by the algorithm. 
Consider a hexagon $ABCDEF$ and suppose that the trajectory of one robot, as long as no exit was found, is $ABDC$. 
Similarly, the other robot follows the symmetric trajectory $FECD$; cf.\ left-hand side of \figurename~\ref{fig:hex_ub}. 
By symmetry it is sufficient to consider exits at vertices $A$, $B$, or $C$. 
An exit at $C$ is reached by each robot independently, while both robots proceed to an exit at $A$ or $B$ after meeting at point $M$, the intersection of the segments $BD$ and $EC$. 
Altogether, they need a total time of at most $\max\{1 + 4/\sqrt(3),~ 1 + (2+\sqrt{7})/\sqrt{3},~ 1+\sqrt{3}+1\}$ to evacuate from the hexagon. 
An interested reader may verify that, in each case, the evacuation time of this algorithm is always upper bounded by $2+\sqrt{3}$.

\begin{figure}[!ht]
                \centering
 \vspace{-2ex}
 \includegraphics[scale=0.7,page=3]{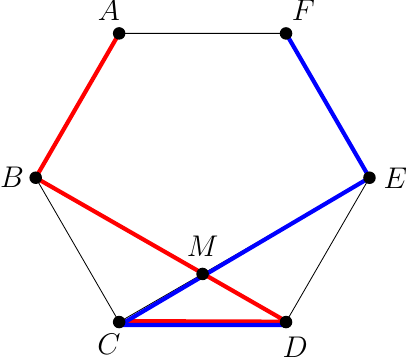} \hspace{1cm}
                \includegraphics[scale=0.7,page=2]{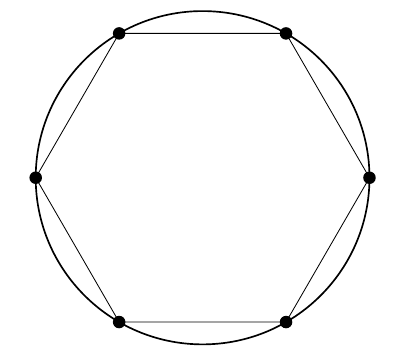}
				\caption{The trajectories for \RA~(red) and \RB~(blue) for the hexagon evacuation algorithm having evacuation time $2+\sqrt{3}$, 					while the exit has not been found, are depicted on the left.
					Right-hand side: At time $1 + \frac{\pi}{6} - \varepsilon$, there is regular hexagon all of whose vertices are unexplored and lie on the boundary of the disk.
				}
				\label{fig:hex_ub}
\end{figure}

In the above algorithm, the robots meet at M, regardless of  whether the exit has been already found or not. 
The idea of our algorithm for disk evacuation presented in the previous section was influenced by this non-intuitive presence of a forced meeting. \\
Combining Theorem~\ref{thm:hexagon} with the fact that hexagon edges correspond to arcs of length $\pi/6$, we obtain the following lower bound for our evacuation problem.

\begin{theorem}
Assume you have a unit disk with an exit placed somewhere on the boundary. 
The worst case evacuation time  for two robots starting at the centre in the \ff model 
  is at least $3+\frac{\pi}{6} + \sqrt{3} \approx 5.255$. 
\end{theorem}
\begin{proof}
It takes 1 time unit for the robots to reach the boundary of the disk.
By time $t=1+\frac{\pi}{6}$,
any algorithm could have explored at most $\frac{2\pi}{6}$ of the boundary of the disk. 
Hence for any $\varepsilon$ with $0 < \varepsilon < t$, there exists a regular hexagon with all vertices on the boundary of the disk and all of whose vertices are unexplored at time $t-\varepsilon$; see the right-hand side of Figure~\ref{fig:hex_ub}. 
As the exit might as well be on any vertex of this hexagon, invoking Theorem~\ref{thm:hexagon} implies a lower bound of $1+\frac{\pi}{6} + 2 + \sqrt{3}$ for the worst case of evacuating both robots. 
\end{proof}

Our lower bound technique could be extended to $n\geq 7$, and it might 
improve the lower bound value. However, the case analysis becomes 
extensive,
very complicated, and it would not provide any more insight into the 
lower bound question.

\section{Conclusion}
\label{sec:Conclusion}

In this paper we studied evacuating two robots from a disk, where the robots 
can collaborate using \ff communication.
Unlike evacuation for two robots in the wireless communication
model, for which the tight bound $1 + \frac{2\pi}{3}+ \sqrt{3}$
is proved in~\cite{CGGKMP}, the evacuation problem for two robots in
the \ff model is much harder to solve.
We gave a new non-trivial algorithm for the  \ff communication
model, by this improving the upper bound of 5.740
in~\cite{CGGKMP} to 
5.628
(the upper bound has since been improved 
to 5.625 in~\cite{Watten2017}
and to 5.6234 in~\cite{disser2019evacuating}
).
We used a novel, non-intuitive idea of a forced meeting between the robots, 
regardless of whether the exit was found before the meeting (although~\cite{Watten2017,disser2019evacuating} have since demonstrated that by avoiding a forced meeting one can improve the upper bound). 
We also provided a different analysis that improves the lower bound in~\cite{CGGKMP}.
Further tightening of the upper and lower bounds remains a challenging open question. Along the same lines, results pertaining to the generalization of our auxiliary problem of evacuating from an $n\textrm{-gon}$ is another interesting open problem.

\acknowledgements
\label{sec:ack}
This work was initiated during the \emph{$13^{th}$ Workshop on Routing} held in July 15--19,
2014 in Quer\'{e}taro, M\'{e}xico.
The authors would also like to thank the anonymous reviewers who, during the journal evaluation process, had a number of suggestions that improved the manuscript. 

\nocite{*}
\bibliographystyle{abbrvnat}
\bibliography{refs-August2019}
\label{sec:biblio}

\end{document}